%% file: reduction4.tex
\documentclass[final,11pt]{amsart}

\usepackage{amsmath,amssymb,amsfonts}
\usepackage{epsfig} \usepackage{latexsym,nicefrac,bbm}
\usepackage{xspace}
\usepackage{color,fancybox,graphicx,url,subfigure,fullpage}
\usepackage[top=0.95in, bottom=0.95in, left=0.95in, right=0.95in]{geometry}
\usepackage{tabularx} \usepackage{hyperref} 
\usepackage{pdfsync}

\newcommand{\defeq}{\stackrel{\textup{def}}{=}}
\newcommand{\nfrac}{\nicefrac}
\newcommand{\eps}{\epsilon}

\renewcommand{\epsilon}{\varepsilon}

\def\showauthornotes{1} 
 
\def\showdraftbox{1}

\input{macros}

\begin{document}
\title{\bf Matrix Inversion Is As Easy As Exponentiation}
\date{}

\begin{abstract}
  We prove that the inverse of a positive-definite matrix can be
  approximated by a weighted-sum of a small number of matrix
  exponentials. Combining this with a previous result \cite{OSV}, we
  establish an equivalence between matrix inversion and exponentiation
  up to polylogarithmic factors.  In particular, this connection
  justifies the use of Laplacian solvers for designing fast
  semi-definite programming based algorithms for certain graph
  problems. The proof relies on the Euler-Maclaurin formula and
  certain bounds derived from the Riemann zeta function.
\end{abstract}

\author{Sushant Sachdeva}
\thanks{Sushant Sachdeva. Department of Computer Science, Princeton University,
USA. \url{sachdeva@cs.princeton.edu}}

\author{Nisheeth K. Vishnoi}
\thanks{Nisheeth K. Vishnoi. Microsoft Research, Bangalore, India. \url{nisheeth.vishnoi@gmail.com}}

\maketitle

\vspace{-6mm}

\section{Matrix Inversion vs. Exponentiation}
Given a symmetric $n \times n$ matrix $A,$ its matrix exponential is
defined to be $e^A \defeq \sum_{i \geq 0} \frac{A^i}{i!}.$ This
operator is of fundamental interest in several areas of mathematics,
physics, and engineering, and has recently found important
applications in algorithms, optimization and quantum
complexity. Roughly, these latter applications are manifestations of
the matrix-multiplicative weight update method and its deployment to
solve semi-definite programs efficiently (see
\cite{AK,OSV,QIPJournal}).  For fast graph algorithms, the quantity of
interest is $e^{-L}v,$ where $L$ is the combinatorial Laplacian of a
graph, and $v$ is a vector. The vector $e^{-L}v$ can also be
interpreted as the resulting distribution of a certain continuous-time
random walk on the graph with starting distribution $v.$ In
\cite{OSV}, appealing to techniques from approximation theory, the
computation of $e^{-L}v$ was reduced to a small number of computations
of the form $L^{-1}u.$ Thus, using the near-linear-time Laplacian
solver\footnote{A Laplacian solver is an algorithm that
  (approximately) solves a given system of linear equations $Lx=b,$
  where $L$ is a graph Laplacian and $b\in {\rm Im}(L)$, \emph{i.e.},
  it (approximately) computes $L^{-1}b,$ see \cite{Vishnoi12}.} due to
Spielman and Teng~\cite{SpielmanT04}, this gives an
$\tilde{O}(m)$-time algorithm for approximating $e^{-L}v$ for graphs
with $m$ edges.  The question of whether the Spielman-Teng result is
necessary in order to compute $e^{-L}v$ in near-linear time remained
open, see \cite[Chapter 9]{Vishnoi12}.  We answer this question in the
affirmative by presenting a reduction in the other direction, again
relying on analytical techniques.  The following is our main result.
\begin{theorem}
\label{thm:matrix-version}
Given $\eps,\delta \in (0,1],$ there exist $\poly( \log
\nfrac{1}{\delta\eps})$ numbers $0< w_j,t_j 
= O(\poly(\nfrac{1}{\delta\eps})),$ such that for all symmetric matrices
$A$ satisfying $\delta I \preceq A \preceq I,$ $(1-\eps)A^{-1}
\preceq \sum_j w_j e^{-t_j A} \preceq (1+ \eps)A^{-1}.$
\end{theorem}
\noindent 
This proves that the problems of matrix exponentiation and matrix
inversion are equivalent up to polylogarithmic factors. This result
justifies the somewhat surprising use of Laplacian solvers for
matrix-exponential based methods for designing fast semi-definite
programming based algorithms for certain graph problems.  Note that
this equivalence does not require the matrix $A$ to be a Laplacian,
but only that it be a symmetric positive-definite matrix. It would be
interesting to investigate if this result can be used to construct
fast solvers for linear systems more general than those arising from
graph Laplacians. Finally, note that the numbers $w_j,t_j$ in the
above theorem are {\em independent} of the matrix $A,$ and are given
explicitly in the proof.

The proof of Theorem~\ref{thm:matrix-version}
follows from the  lemma below, which
gives such an approximation in the scalar world.
\begin{lemma}
\label{lem:exp-approx}
Given $\eps,\delta \in (0,1],$ there exist $\poly( \log
\nfrac{1}{\delta\eps})$ numbers $0< w_j,t_j =
O(\poly(\nfrac{1}{\delta\eps})),$ such that for all $x \in
[\delta,1],$ $(1-\eps)x^{-1} \le \sum_j w_j e^{-t_j x} \le (1+\eps)
x^{-1}.$
\end{lemma}

\noindent
Note that as $x$ approaches $0$ from the right, $x^{-1}$ is unbounded,
where as $e^{-t x}$ is bounded by 1 for any $t > 0.$ This justifies
the assumption that $x \in [\delta,1].$ Versions of this lemma were
proved in~\cite{BM05,BM10}. Our proof is simple and self-contained. We
attempt to make the analytical techniques used in the proofs
accessible to a wider theory audience.  We begin by showing how Lemma
\ref{lem:exp-approx} implies Theorem~\ref{thm:matrix-version}.

\subsection*{Proof of Theorem~\ref{thm:matrix-version}}
Let $\{\lambda_i\}_i$ be the eigenvalues of $A$ with corresponding
eigenvectors $\{u_i\}_i.$ Since $A$ is symmetric and $\delta I \preceq
A \preceq I,$ we have $\lambda_i \in [\delta,1],$ for all $i.$ Let
$w_j,t_j > 0$ denote the numbers given by Lemma~\ref{lem:exp-approx}
for parameters $\eps$ and $\delta.$ Thus, Lemma~\ref{lem:exp-approx}
implies that all $i,$ $(1-\eps)\lambda_i^{-1} \le \sum_j w_j e^{-t_j
  \lambda_i} \le (1+\eps) \lambda_i^{-1}.$ Note that if $\lambda_i$ is
an eigenvalue of $A,$ then $\lambda_i^{-1}$ is the corresponding
eigenvalue of $A^{-1}$ and $e^{-t_j \lambda_i}$ is that of $e^{-t_jA}$
with the same eigenvector.  Thus, multiplying the scalar inequalities
by $u_iu_i^\top$ and summing up, we obtain the matrix inequality
$(1-\eps) \sum_{i} \lambda_i^{-1} u_iu_i^\top \preceq \sum_j w_j
\sum_{i} e^{-t_j\lambda_i} u_iu_i^\top \preceq (1+\eps) \sum_{i}
\lambda_i^{-1} u_iu_i^\top.$ Hence, $ (1-\eps)A^{-1} \preceq \sum_j
w_j e^{-t_j A} \preceq (1+\eps) A^{-1}.$

\subsection{Integral Representation, Discretization  and Smoothness}
\label{sec:outline}
The starting point of the proof of Lemma \ref{lem:exp-approx} is the
easy integral identity $x^{-1} = \int_0^{\infty} e^{-xt} dt.$ Thus, by
discretizing this integral to a sum, the fact that one can approximate
$x^{-1}$ as a weighted sum of exponentials as claimed
Lemma~\ref{lem:exp-approx} is not surprising. The crux is to prove
that this can be achieved using a \emph{sparse} sum of exponentials.
One way to discretize an integral to a sum is the so called
\emph{trapezoidal rule}.
If $g$ is the integrand, and $[a,b]$ is the interval of integration,
this rule approximates the integral $\int_a^b g(t)dt$ by covering the
area under $g$ in the interval $[a,b]$ using \emph{trapezoids} of
small width, say $h,$ as follows: 
\vspace{-2mm}
\[\int_a^b g(t)dt \ \approx\  T_{g}^{[a,b],h} \ \defeq \ \frac{h}{2}
\cdot \sum_{j=0}^{K-1} \left(g(a+jh)+g(a+(j+1)h) \right),\] where $K
\defeq \frac{b-a}{h}$ is an integer. The choice of $h$ determines the
discretization of the interval $[a,b],$ and hence $K$, which is
essentially the sparsity of the approximating sum.  To apply this to
the integral representation for $x^{-1},$ we have to first truncate
the infinite integral $\int_{0}^\infty e^{-xt}dt$ to a large enough
interval $[0,b],$ and then bound the error in the trapezoidal rule.
Recall that the error needs to be of the form
\[\textstyle \abs{x^{-1} - \frac{h}{2}\sum_{j} \left(e^{-xjh} +
    e^{-x(j+1)h} \right)} \le \eps x^{-1}.\] For such an error
guarantee to hold, we must have $xh \le O_{\eps}(1).$ Thus, if we want
the approximation to hold for all $0 < x \leq 1,$ we require $h \le
O_\eps(1),$ which in turn implies that $K \geq \Omega_\eps({b}).$
Also, if we restrict the interval to $[0,b],$ the truncation error is
$\int_b^{\infty} e^{-xt} dt = x^{-1} e^{-bx},$ forcing $b \ge
\delta^{-1}\log \nfrac{1}{\eps}$ for this error to be at most
$\nfrac{\eps}{x}$ for all $x \in [\delta,1].$ Thus, this way of
discretizing can only give us a sum which uses
poly$\left(\nfrac{1}{\delta}\right)$ exponentials, which does not
suffice for our application.

This suggests that we should pick a discretization such that $t,$
instead of increasing linearly with $h,$ increases much more
rapidly. Thus, a natural idea is to allow $t$ to grow
geometrically. This can be achieved by substituting $t=e^s$ in the
above integral to obtain the identity $x^{-1} = \int_{-\infty}^\infty
e^{-xe^s + s} ds.$ We show that discretizing this integral using the
trapezoidal rule does indeed give us the lemma.

For convenience, we define $f_x(s) \defeq e^{-xe^s +s}.$ First,
observe that $f_x(s) = x^{-1}\cdot f_1(s+\ln x).$ Since we also allow
the error to scale as $x^{-1}$, as $x$ varies over $[\delta,1],$ $s$
needs to change only by an additive $\log \nfrac{1}{\delta}$ to
compensate for $x.$ Roughly, this suggests that when approximating
this integral by the trapezoidal rule, the dependence on
$\nfrac{1}{\delta}$ is likely logarithmic, instead of polynomial. The
proof formalizes this intuition and uses the fact that the error in
the approximation by the trapezoidal rule can be expressed using the
Euler-Maclaurin formula (see Section \ref{sec:EM}) which involves
higher order derivatives of $f_x.$ We establish the following
properties about the derivatives of $f_x$ which, when combined with
known estimates on Bernoulli numbers obtained from the Riemann zeta
function, allow us to bound this error with relative ease (see
Section~\ref{sec:infinite-exp}): (1) All the derivatives of $f_x$ up
to any fixed order, vanish at the end points of the integration
interval (in the limit). (2) The derivatives of $f_x$ are reasonably
\emph{smooth}; the $L_1$ norm of the $k$-th derivative is bounded
roughly by $x^{-1} k^k$ (see Lemma~\ref{lem:smoothness}). In summary,
this allows us to approximate $x^{-1}$ as an {\em infinite} sum of
exponentials. In this sum, the contribution beyond about $\poly(\log
{\nfrac{1}{\eps \delta}})$ terms turns out to be negligible, and hence
we can truncate the infinite sum to obtain our final approximation
(see Section~\ref{sec:truncate}).
 
We now present simple properties of the derivatives of $f_x,$ alluded
to above, which underlie the technical intuition as to why an
approximation of the kind claimed in Lemma~\ref{lem:exp-approx} should
exist. Let $f_x^{(k)}(s)$ denote the $k^\textrm{th}$ derivative of the
function $f_x$ with respect to $s$. The first fact relates $f_x
^{(k)}(s)$ to $f_x(s).$
\begin{fact}
\label{lem:derivative}
 For any non-negative integer $k,$  $f^{(k)}_x(s) = f_x(s)
  \sum_{j=0}^k c_{k,j} (-xe^s)^{j} ,$ where $c_{k,j}$ are some
  non-negative integers
  satisfying $\sum_{j=0}^k c_{k,j} \le (k+1)^{k+1}.$
\end{fact}
\begin{proof}
We prove this lemma by induction on $k$. For $k=0,$ we have
$f_x^{(0)}(s) = f_x(s).$ Hence, $f_x^{(0)}$ is of the required form,
with $c_{0,0} = 1,$ and $\sum_{j=0}^0 c_{0,j} = 1.$ Hence, the claim
holds for $k=0.$ Suppose the claim holds for $k.$ Hence, $f^{(k)}_x(s)
= f_x(s)\sum_{j=0}^k c_{k,j} (-xe^s)^{j},$ where $c_{k,j}$ are
non-negative integers satisfying $\sum_{j=0}^k c_{k,j} \le
(k+1)^{k+1}.$ We can compute $f_x^{(k+1)}(s)$ as follows,
\begin{align*}
f_x^{(k+1)}(s) & = \frac{d}{ds} \left( \sum_{j=0}^k c_{k,j}
  (-xe^s)^{j} f_x(s)\right) = \sum_{j=0}^k c_{k,j}
  (j -xe^s + 1) (-xe^s)^{j} f_x(s) \\
& = f_x(s) \sum_{j=0}^{k+1} ((j+1)c_{k,j} + c_{k,j-1}) (-xe^s)^j,
\end{align*}
where we define $c_{k,k+1}\defeq 0,$ and $c_{k,-1} \defeq 0.$ Thus, if
we define $c_{k+1,j} \defeq (j+1) c_{k,j} + c_{k,j-1},$ we get that
$c_{k+1,j} \ge 0,$ and that $f_x^{(k+1)}$ is of the required
form. Moreover, we get, $\sum_{j=0}^{k+1} c_{k+1,j} \le
(k+2)(k+1)^{k+1} +(k+1)^{k+1} = (k+3)(k+1)(k+1)^{k} \le (k+2)^2(k+1)^k
\le (k+2)^{k+2}.$ This proves the claim for $k+1$ and, hence, the fact
follows by induction.
\end{proof}

\noindent
The next lemma uses the fact above to bound the $L_1$ norm of
$f^{(k)}_x.$
\begin{lemma}
\label{lem:smoothness}
For every non-negative integer $k,$ $\int_{-\infty}^\infty
\abs{f_x^{(k)}(s)}ds \le \frac{2}{x} \cdot {e^{k}(k+1)^{2k}}.$
\end{lemma}
\begin{proof} By Fact \ref{lem:derivative}, $\int_{-\infty}^\infty
  \abs{f_x^{(k)}(s)}ds$ is at most
\begin{align*}
  & \int_{-\infty}^\infty \abs{\left( \sum_{j=0}^{k} c_{k,j}
      (-xe^s)^{j} \right)} e^{-xe^s +s } ds \stackrel{t=xe^s}{=}
  \frac{1}{x} \int_{0}^\infty \abs{\left(
      \sum_{j=0}^{n} c_{k,j} (-t)^{j} \right)} e^{-t} dt \\
  & \qquad \stackrel{{\rm Fact}\; \ref{lem:derivative}}{\le}
  \frac{1}{x} (k+1)^{k+1} \left(\int_0^1 e^{-t} dt + \int_1^\infty
    t^{k} e^{-t} dt\right) \le \frac{1}{x} \cdot {(k+1)^{k+1} \cdot
    (1+k!)} \le \frac{2}{x} \cdot {e^{k}(k+1)^{2k}},
\end{align*}
where the last inequality uses $k+1 \le e^{k},$ and $1+k! \le 2(k+1)^{k}.$
\end{proof}

We conclude this section by giving a brief comparison of our proof to
that from \cite{BM05}. While the authors in \cite{BM05} employ both
the trapezoidal rule and the Euler-Maclaurin formula, our proof
strategy is different and leads to a shorter and simpler proof. In
contrast to the previous proof, we use the Euler-Maclaurin formula in
the limit over $[-\infty,\infty],$ and since the derivatives of $f_x$
vanish in the limit, we save considerable effort in bounding the
derivatives at the end points of the integral, which is required when
using the Euler-Maclaurin formula to bound the error. We manage to use
simpler bounds, at the cost of slightly worse parameters. On the way,
we obtain an approximation of $x^{-1}$ as an infinite sum of
exponentials that holds for all $x > 0,$ which we believe is interesting
in itself.

\section{Proof of Lemma \ref{lem:exp-approx}}
Before we introduce the Euler-Maclaurin formula which captures the
error in the approximation of an integral by the trapezoidal rule, we
introduce the Bernoulli numbers and polynomials, bounds on which are
derived using a connection to the Riemann zeta function.

\subsection{Bernoulli Polynomials and Euler-Maclaurin Formula}
\label{sec:EM}
The Bernoulli numbers, denoted by $b_k$ for any integer $k \ge 0,$ are
a sequence of rational numbers which, while discovered in an attempt
to compute sums of the form $\sum_{i \geq 0}^n i^k,$ have deep
connections to several areas of mathematics, including number theory
and analysis.\footnote{The story goes that when Charles Babbage
  designed the Analytical Engine in the 19th century, one of the most
  important tasks he hoped the Engine would perform was the
  calculation of Bernoulli numbers.} They can be defined recursively
as: $b_0 = 1,$ and the following equation which is satisfied for all
positive integers $k \ge 2,$ $ \sum_{j=0}^{k-1} {k \choose j} b_j= 0.
$ This implies that $(e^t-1) \sum_{k=0}^\infty b_k \frac{t^k}{k!} =
t.$ Further, it can be checked that $\frac{t}{2} + \frac{t}{e^t-1}$ is
an even function, thus implying that $b_k=0$ for odd $k \ge 2.$ Given
the Bernoulli numbers, the Bernoulli polynomials are defined as
$B_k(s) \defeq \sum_{j=0}^k \binom{k}{j} b_j s^{k-j}.  $ It follows
from the definition that, for all $k,$ and $j \le k,$
$\frac{B_k^{(j)}(s)}{k!}  = \frac{B_{k-j}(s)}{(k-j)!}.$ We also get
$B_0(s) \equiv 1,$ $B_1(s) \equiv s - \frac{1}{2}.$ Moreover, using
the definition of Bernoulli numbers, we get that $B_k(0) = B_k(1) =
b_k$ for all $k \ge 2.$ We also need the following bounds on the
Bernoulli polynomials and the Bernoulli numbers.
\begin{lemma}
\label{lem:bernoulli}
For any non-negative integer $k,$ and for all $s \in [0,1],$ 
$\frac{|B_{2k}(s)|}{(2k)!} \le \frac{\abs{b_{2k}}}{(2k)!} \le \frac{4}{(2\pi)^{2k}}.$
\end{lemma}
\begin{proof}
  The first inequality follows from a well-known fact that
  $|B_{2k}(s)| \leq |b_{2k}|$ for all $s \in [0,1]$ (see
  \cite{Graham:1994:CMF:562056}). For the second inequality, we recall
  the following connection between Bernoulli numbers and the Riemann
  zeta function for any even positive integer, proved by Euler (see
  \cite{Graham:1994:CMF:562056}), $\zeta(2k) \defeq \sum_{j \ge 1}
  {j^{-2k}} = (-1)^{k+1}\frac{b_{2k} (2\pi)^{2k}}{2 \cdot (2k)!}.$
  Thus, $\frac{\abs{b_{2k}}}{(2k)!} = \frac{2}{(2\pi)^{2k}} \sum_{j
    \ge 1} j^{-2k} \le 4(2\pi)^{-2k}.$
\end{proof}

\renewcommand{\Re}{\mathbb{R}}

\noindent
One of the most significant connections in analysis involving the
Bernoulli numbers is the \emph{Euler-Maclaurin formula} which
describes the error in approximating an integral by the trapezoidal
rule.
\begin{lemma}[Euler-Maclaurin Formula]
  Given a function $g : \Re \to \Re,$ for any $a < b,$ any positive
  $h,$ and any positive integer $N \in \nat,$ we have,
\begin{align}
\label{eq:EM}
  \int_a^b g(s)ds - T_{g}^{[a,b],h} = h^{2N+1} \int_0^K
  \frac{B_{2N}(s-[s])}{(2N)!} g^{(2N)}(a+sh) ds - \sum_{j=1}^N
  \frac{b_{2j}}{(2j)!}h^{2j} \left(g^{(2j-1)}(b)-g^{(2j-1)}(a)\right),
\end{align}
where $K \defeq \frac{b-a}{h}$ is an integer, and $[\cdot]$ denotes
the integer part.
\end{lemma}
\noindent
Note that the Euler-Maclaurin formula is really a family of formulae,
one each for the choice of $N,$ which we call the \emph{order} of the
formula. Also note that this formula captures the error {\em exactly}. This
error can be much less than the naive bound obtained by summing up the
absolute value of the error due to each trapezoid.  The first term in
~\eqref{eq:EM}, after removing the contribution due to the Bernoulli
polynomials via Lemma \ref{lem:bernoulli}, can be bounded by the $L_1$
norm of $g^{(2N)}.$ The second term in \eqref{eq:EM} depends only on
$g^{(2N-1)}$ evaluated at the ends of the interval.  The choice of $N$
is influenced by how well behaved the higher order derivatives of the
function are. For example, if $g(s)$ is a polynomial, when $2N >
\text{degree}(g),$ we get an exact expression for $\int_a^b g(s)ds$ in
terms of the values of the derivatives of $g$ at $a$ and $b.$

In the next section, we use the Euler-Maclaurin formula to bound the
error in approximating the integral $\int f_x(s) ds$ using the
trapezoidal rule. For our application, we pick $a$ and $b$ such that
the derivatives up to order $2N-1$ at $a$ and $b$ are
negligible. Since the sparsity of the approximation is
$\Omega(\nfrac{1}{h}),$ for the sparsity to depend logarithmically on
the error parameter $\eps,$ we need to pick $N$ to be roughly
$\Omega(\log \nfrac{1}{\eps}),$ so that the first error term in
\eqref{eq:EM} is comparable to $\eps.$

We end this section by giving a proof sketch for the Euler-Maclaurin
formula (see also \cite{tao}). By a change of variables, it suffices to
prove the formula for $h=1$ and for the interval $[0,1].$ Consider the
integral $\int_0^1 \frac{B_{2N}^{(2N)}(s)}{(2N)!}g(s)ds,$ and apply
integration by parts\footnote{$\int \frac{du}{ds}v ds = u v - \int u
  \frac{dv}{ds} ds.$} to it repeatedly to obtain
\begin{align*}
  \int_0^1 \frac{B^{(2N)}_{2N}(s)}{(2N)!} g(s) ds = &
  \left. \frac{B_{2N}^{(2N-1)}(s)}{(2N)!} g(s) \right\rvert_0^1 -
  \left. \frac{B_{2N}^{(2N-2)}(s)}{(2N)!} g^{(1)}(s) \right\rvert_0^1
  + \left. \frac{B_{2N}^{(2N-3)}(s)}{(2N)!} g^{(2)}(s) \right\rvert_0^1
 \\
  & - \cdots - \left. \frac{B_{2N}(s)}{(2N)!} g^{(2N-1)}(s) \right\rvert_0^1 +
\int_0^1 \frac{B_{2N}(s)}{(2N)!} g^{(2N)}(s) ds.
\end{align*}

\noindent
Using the fact that for all $k \le 2N,$ $\frac{B_{2N}^{(k)}(s)}{(2N)!} =
\frac{B_{2N-k}(s)}{(2N-k)!},$ and rearranging, we get,
\begin{align*}
  \int_0^1 {B_{0}(s)} g(s) ds - \left. {B_{1}(s)} g(s)
    \vphantom{\frac{B_1(s)}{b!}} \right\rvert_0^1 = \sum_{k=2}^{2N}
  (-1)^{k-1}\left. \frac{B_{k}(s)}{k!} g^{(k-1)}(s) \right\rvert_0^1 +
  \int_0^1 \frac{B_{2N}(s)}{(2N)!} g^{(2N)}(s) ds.
\end{align*}
\noindent
Now, using $B_0(s) \equiv 1,$ we get that the first term on the
l.h.s. is $\int_0^1 g(s) ds.$ Also, since $B_1(1) = \nfrac{1}{2},
B_1(0) = -\nfrac{1}{2},$ we see that the second term on the l.h.s. is
$\nfrac{1}{2}\cdot (g(0)+g(1)) = T_g^{[0,1],1}.$ Finally, using
$B_k(0) = B_k(1) =b_k$ for $k \ge 2,$ and that $b_k = 0$ when $k \ge
2$ is odd, we get the desired formula.

\subsection{Approximation Using an Infinite Sum}
\label{sec:infinite-exp}
As mentioned in Section~\ref{sec:outline}, we approximate the integral
$\int_{-\infty}^{\infty} f_x(s)ds$ using the trapezoidal rule. We
bound the error in this approximation using the Euler-Maclaurin
formula. Since the Euler-Maclaurin formula applies to finite
intervals, we first fix the step size $h,$ use the Euler-Maclaurin
formula to bound the error in the approximation over the interval
$[-bh,bh]$ (where $b$ is some positive integer), and then let $b$ go
to $\infty.$ This allows us to approximate the integral over
$[-\infty,\infty]$ by an infinite sum of exponentials. In the next
section, we truncate this sum to obtain our final approximation.

We are given $\eps,\delta \in (0,1].$ Fix an $x \in [\delta,1],$ the
step size $h=\Theta\left((\log \nfrac{1}{\eps})^{-2}\right),$ and the
order of the Euler-Maclaurin formula, $N=\Theta\left(\log
  \nfrac{1}{\eps}\right)$ (exact parameters to be specified
later). For any positive integer $b,$ applying the order $N$
Euler-Maclaurin formula to the integral $\int_{-bh}^{bh} f_x(s) ds,$
and using bounds from Lemma~\ref{lem:bernoulli}, we get,
\begin{align}
\label{eq:EM-error}
\abs{\int_{-bh}^{bh} f_x(s)ds - T_{f_x}^{[-bh,bh],h}} \le & 4\left(
  \frac{h}{2\pi}
\right)^{2N} \int_{-bh}^{bh} \abs{f^{(2N)}_x(s)} ds \\
& \qquad + \sum_{j=1}^N 4 \left( \frac{h}{2\pi}\right)^{2j}
\left(\abs{f^{(2j-1)}_x(-bh)}
  +\abs{f^{(2j-1)}_x(bh)}\right). \nonumber
\end{align}
Now, we can use Fact~\ref{lem:derivative} to bound the derivatives in
the last term of
~\eqref{eq:EM-error}. Fact~\ref{lem:derivative} implies that
for any $s$ and any positive integer $k,$ $\abs{f^{(k)}(s)} \le
f_x(s)(k+1)^{k+1} \max\{1,(xe^s)^{k}\}.$ Thus, for $b \ge
-\frac{1}{h}\log \frac{1}{x},$ we have $xe^{-bh} \le 1$ and
$\abs{f^{(k)}(-bh)} \le e^{-bh} (k+1)^{k+1},$ and hence $f^{(k)}(-bh)$
vanishes for any fixed $k$ and $h,$ as $b$ goes to $\infty.$ Also, for
any $x > 0,$ and $b > \frac{1}{h} \log \frac{1}{x},$ we get,
$\abs{f^{(k)}(bh)} \le x^ke^{(k+1)bh-xe^{bh}} (k+1)^{k+1},$ which again
vanishes for any fixed $k$ and $h,$ as $b$ goes to $\infty.$ Thus,
letting $b$ go to $\infty$ and observing that $T_{f_x}^{[-bh,bh],h}$
converges to $h \sum_{j \in \bz} f_x(jh),$ ~\eqref{eq:EM-error}
implies,
\begin{align}
\label{eq:inf-sum-error}
\abs{\int_{-\infty}^\infty f_x(s)ds - h\sum_{j \in \bz} f_x(jh)} \le
4\left( \frac{h}{2\pi} \right)^{2N} \int_{-\infty}^{\infty}
\abs{f^{(2N)}_x(s)} ds.
\end{align}
Hence, since the derivatives of the function $f_x(s)$ vanish as $s$
goes to $\pm \infty,$ the error in approximating the integral over
$[-\infty,\infty]$ is just controlled by its \emph{smoothness}. Since
we already know $f_x$ is a very smooth function, we are in good shape.
Using Lemma~\ref{lem:smoothness}, we get, $
\left(\frac{h}{2\pi}\right)^{2N} \int_{-\infty}^{\infty}
\abs{f^{(2N)}_x(s)} ds \le
\frac{2}{x}\left(\frac{(2N+1)^2eh}{2\pi}\right)^{2N}.$ Thus, if we let
$h \defeq \frac{2\pi}{e^2(2N+1)^2},$ and $N \defeq \ceil{
  \frac{1}{2}\log \frac{24}{\eps}},$ ~\eqref{eq:inf-sum-error} implies
that,
\begin{align}
\label{eq:inf-sum-error-final}
\abs{x^{-1} - h\sum_{j \in \bz} e^{jh}\cdot e^{-xe^{jh}}} =
\abs{\int_{-\infty}^\infty f_x(s)ds - h\sum_{j \in \bz} f_x(jh)} \le
8e^{-2N}\cdot \frac{1}{x} \le \frac{\eps}{3} \frac{1}{x}.
\end{align}
Also note that the above approximation holds for all $x > 0.$ Thus, in
particular, we can approximate the function $x^{-1}$ over $[\delta,1]$
as an (infinite) sum of exponentials.

\subsection{Truncating the Infinite Sum and Proof of Lemma \ref{lem:exp-approx}}
\label{sec:truncate}
Now, we want to truncate the infinite sum of exponentials
approximating $x^{-1}$ given by
~\eqref{eq:inf-sum-error-final}. Since the function
$f_x(s)=e^{s}\cdot e^{-xe^{s}}$ is non-decreasing for $s < \log
\nfrac{1}{x},$ we  majorize the lower tail by an integral.  For $A
\defeq \floor{-\frac{1}{h}\log \frac{3}{\eps}} < 0 \le \frac{1}{h}\log
\frac{1}{x}$ (since $x \le 1$),

\begin{align}
\label{eq:lower-tail}
h\sum_{j < A} e^{jh}\cdot e^{-xe^{jh}} \le h \int_{-\infty}^A
e^{jh}\cdot e^{-xe^{jh}} dj = \int_{0}^{e^{Ah}} e^{-xt} dt = x^{-1}
\left(1 - e^{-xe^{Ah}}\right) \le \frac{\eps}{3} \frac{1}{x}.
\end{align}

\noindent
Again, for the upper tail, since the function $f_x(s) = e^{s}\cdot
e^{-xe^{s}}$ is non-increasing for $s \ge \log \frac{1}{x},$ we
majorize by an integral. For $B \defeq \ceil{\frac{1}{h} \log \left(
    \frac{1}{\delta} \log \frac{3}{\eps} \right)} \ge \frac{1}{h}\log
\frac{1}{x}$ (since $x \ge \delta$ and $\eps \le 1$),
\begin{align}
\label{eq:upper-tail}
h\sum_{j > B} e^{jh}\cdot e^{-xe^{jh}} \le h \int_{B}^\infty
e^{jh}\cdot e^{-xe^{jh}} dj = \int_{e^{Bh}}^{\infty} e^{-xt} dt =
x^{-1} \cdot e^{-xe^{Bh}} \le \frac{\eps}{3} \frac{1}{x}.
\end{align}
Before we complete the proof, we list here the setting of all
parameters for completeness:
\[N = \ceil{\frac{1}{2}\log \frac{24}{\eps}},\ h =
\frac{2\pi}{e^2(2N+1)^2},\ A = \floor{-\frac{1}{h} \log
  \frac{3}{\eps}},\ B = \ceil{\frac{1}{h} \log \left( \frac{1}{\delta}
    \log \frac{3}{\eps} \right)}.\] Thus, combining
~\eqref{eq:inf-sum-error-final},
\eqref{eq:lower-tail}~and~\eqref{eq:upper-tail}, the final error is
given by,
\begin{align*}
  \abs{\frac{1}{x}-h\sum_{j \ge A}^B e^{jh}\cdot e^{-xe^{jh}} } \le &
  \abs{\frac{1}{x} - h\sum_{j \in \bz} e^{jh}\cdot e^{-xe^{jh}} } + h
  \sum_{j < A} e^{jh}\cdot e^{-xe^{jh}} + h \sum_{j > B} e^{jh}\cdot
  e^{-xe^{jh}} \le \frac{\eps}{x}.
\end{align*}
Hence, $(1-\eps)x^{-1} \le \sum_{j \ge A}^B he^{jh}\cdot e^{-xe^{jh}} \le
(1+\eps)x^{-1},$ implying the claim of Lemma~\ref{lem:exp-approx}.

\bibliographystyle{plain}
\bibliography{reduction}
\end{document}

%% file: macros.tex
\newtheorem{theorem}{Theorem}[section]

\newtheorem{lemma}[theorem]{Lemma}

\newtheorem{fact}[theorem]{Fact}

\def\FullBox{\hbox{\vrule width 6pt height 6pt depth 0pt}}

\def\qed{\ifmmode\qquad\FullBox\else{\unskip\nobreak\hfil
\penalty50\hskip1em\null\nobreak\hfil\FullBox
\parfillskip=0pt\finalhyphendemerits=0\endgraf}\fi}

\def\qedsketch{\ifmmode\Box\else{\unskip\nobreak\hfil
\penalty50\hskip1em\null\nobreak\hfil$\Box$
\parfillskip=0pt\finalhyphendemerits=0\endgraf}\fi}

\newcommand\bz{\mathbb Z}
\newcommand\nat{\mathbb N}

\newcommand{\marginlabel}[1]%
{\mbox{}\marginpar{\it{\raggedleft\hspace{0pt}#1}}}
\newcommand{\poly}{\mathrm{poly}}
\newcommand{\floor}[1]{\left\lfloor\, {#1}\,\right\rfloor}
\newcommand{\ceil}[1]{\left\lceil\, {#1}\,\right\rceil}

\definecolor{Mygray}{gray}{0.8}

 \ifcsname ifcommentflag\endcsname\else
  \expandafter\let\csname ifcommentflag\expandafter\endcsname
                  \csname iffalse\endcsname
\fi

\ifnum\showauthornotes=1

\else

\fi

\ifnum\showauthornotes=1
\newcommand{\Authornote}[2]{{\sf\small\color{red}{[#1: #2]}}}
\newcommand{\Authoredit}[2]{{\sf\small\color{red}{[#1]}\color{blue}{#2}}}
\newcommand{\Authorcomment}[2]{{\sf \small\color{gray}{[#1: #2]}}}
\newcommand{\Authorfnote}[2]{\footnote{\color{red}{#1: #2}}}
\newcommand{\Authorfixme}[1]{\Authornote{#1}{\textbf{??}}}
\newcommand{\Authormarginmark}[1]{\marginpar{\textcolor{red}{\fbox{
#1:!}}}}
\else
\newcommand{\Authornote}[2]{}
\newcommand{\Authoredit}[2]{}
\newcommand{\Authorcomment}[2]{}
\newcommand{\Authorfnote}[2]{}
\newcommand{\Authorfixme}[1]{}
\newcommand{\Authormarginmark}[1]{}
\fi

\def\abs#1{\left| #1 \right|}

\newlength{\pgmtab}  
 {
	\begin{enumerate}}{\end{enumerate}}

\newlength{\tpush}
\ifnum\showdraftbox=1

\else

\fi